\documentclass[conference]{IEEEtran}
\usepackage{cite}
\usepackage{amsmath,amssymb,amsfonts}
\usepackage{graphicx}
\usepackage{textcomp}
\usepackage{xcolor}
\usepackage{amsthm}
\usepackage{bm}
\usepackage{algorithm}
\usepackage{algpseudocode}

\def\BibTeX{{\rm B\kern-.05em{\sc i\kern-.025em b}\kern-.08em
    T\kern-.1667em\lower.7ex\hbox{E}\kern-.125emX}}
\begin{document}

\title{An Improved Linear Extractable Sketch Data Structure for Flow Count Statistics}

\author{\IEEEauthorblockN{Patthadon Tantiameorn}
\IEEEauthorblockA{\textit{Department of Computer Engineering} \\
\textit{Kasetsart University}\\
Bangkok, Thailand \\
patthadon.t@ku.th}
\and
\IEEEauthorblockN{Grittin Nuntasombat}
\IEEEauthorblockA{\textit{Department of Computer Engineering} \\
\textit{Kasetsart University}\\
Bangkok, Thailand \\
grittin.n@ku.th}
\and
\IEEEauthorblockN{Jittat Fakcharoenphol}
\IEEEauthorblockA{\textit{Department of Computer Engineering} \\
\textit{Kasetsart University}\\
Bangkok, Thailand \\
jittat@gmail.com}
}

\newcommand{\buck}{\mathcal B}
\newcommand{\Zp}{\mathbb{Z}_p}

\newtheorem{theorem}{Theorem}
\newtheorem{lemma}{Lemma}

\IEEEoverridecommandlockouts

\IEEEpubid{\makebox[\columnwidth]{\ \ \hfill} \hspace{\columnsep}\makebox[\columnwidth]{ }}

\maketitle

\begin{abstract}
Sketch data structures are very useful for computing statistics on streaming data, including network traffic, server requests, and financial transactions.   In recent work, FermatSketch was introduced as an underlying data structure used to monitor changes in network states.  It is a linear data structure that maintains an associated array of counters and supports listing all key-counter pairs while using almost linear space.  Because it is linear, it can be used to monitor changes between two streams with space proportional to the number of items that change.  The data structure is based on a hash table, and all key-counter pairs can be successfully listed when there are slots in the table with exactly one key hashed to them.  We show how to relax this requirement by using additional computational resources when listing the key-counter pairs, thereby improving space efficiency with only a small overhead when collecting statistics.   We achieve this by storing, for each bucket, multiple linear combinations of the counters whose coefficients are generated from the keys.  With this information, certain linear systems can be solved to obtain the key-counter pairs.  A preliminary experiment shows a significant reduction of memory needed for the data structure.  Our work can be viewed as a trade-off between space and time.
\end{abstract}

\begin{IEEEkeywords}
Sketches, data structures, linearity, network measurement
\end{IEEEkeywords}

\section{Introduction}

Data structures that use memory efficiently are crucial for system measurements, as equipment storages are typically limited and hard to expand, unlike that of personal computing devices.

In this paper, we consider a data structure design problem in a network monitoring setting where we would like to track the packet losses.  
While standard techniques (see e.g.,~\cite{LiMKY16-flowradar,LiMKY16-lossradar}) may require space proportional to the total traffic or proportional to the number of distinct flows, the FermatSketch data structure proposed by Kaicheng~{\em et al.}~\cite{ChameleMon23} uses space only proportional to the number of distinct lost flows (which can be much smaller than the total number of flows).  

The crucial design decision of the FermatSketch is that it is a linear data structure; therefore, the difference of two sketches is still a sketch and when lost flows are considered, the difference sketch has non-zero counters only for the lost flows (referred to later as victim flows), see Figure~\ref{fig:interconnect}.
Moreover, this implies that the data structures only need the space large enough to keep the victim flows.
The FermatSketch is the key to the design of ChameleMon~\cite{ChameleMon23}, a network measurement system that supports flow-level measurement, and is also an underlying data structure in the DaVinci Sketch~\cite{Wang25-DaVinci}.  
The FermatSketch data structure follows the idea of invertible bloom lookup tables~\cite{GoodrichM11-ibt,MizrahiBLYR23-ibt-guarantees,FleischhackerLOS-ESA24-ibt-lessmem} introduced by Goodrich and Mitzenmacher~\cite{GoodrichM11-ibt} with a crucial modification to handle count statistics.

\begin{figure}
\begin{center}
\includegraphics[width=0.4\textwidth]{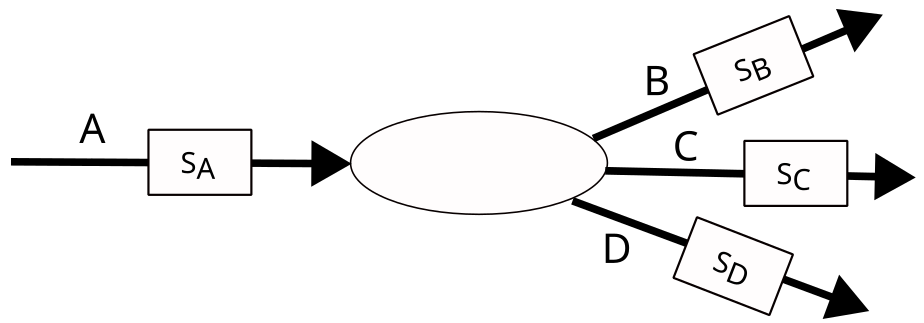}
\end{center}
\caption{We can use linear sketch data structures to track lost flows.  
Linear sketches $S_A,S_B,S_C,$ and $S_D$ track flows on links $A,B,C,$ and $D$, respectively.  To compute lost flows, we can find the difference $S_A-(S_B+S_C+S_D)$.}
\label{fig:interconnect}
\end{figure}

In this work, we present an improvement on the space usage of FermatSketch by a flow extraction procedure that works on relaxed conditions.  In turn, the flow extraction procedure may run much slower.  This can be viewed as a space-time trade-off for data structures that maintain flow counters and support flow counter listing, like the FermatSketch.  We note that while it takes longer to extract flow statistics, we believe that the compute needed can be managed much more efficiently and effectively in centralized infrastructure than the limited memory available on networking devices which are more distributed.

Our main idea is to let each bucket keep multiple linear combinations of the values stored as in the FermatSketch.  
The linear combinations are formed using coefficients that are generated from the keys.
This information is enough to solve the inverse problem using simply Gaussian elimination, if the coefficients are known.  We design the procedure so that the coefficients are selected from a small set of values; thus, by brute force searching through all possible coefficient choices, we can find the correct coefficients and solve the linear system.  
A preliminary experiment presented in Section~\ref{sect:experiments} shows at least $5.8\%$ reduction in memory usage for the $99\%$ accuracy level.   When the number of lost flows is small, the advantage is more significant (i.e., when the number of lost flows is $100$, we get $34.9\%$ reduction in memory usage).  As our data structure can be parameterized to use more time, we believe that we can even further reduce the memory usage.

We formally describe the problem in Section~\ref{sect:prelim}.  The key previous work, the FermatSketch data structure~\cite{ChameleMon23} is also reviewed in that section.  We describe our improvement in Section~\ref{sect:algorithms} and analyze its correctness and its running time requirements in Section~\ref{sect:analysis}.  
Finally, Section~\ref{sect:experiments} shows the experimental results.

\subsection{Related Work}

Sketch data structures~\cite{cormode23podssurvey} enable memory efficient building blocks for many applications.  
Among them, CountMin Sketch~\cite{cormodeM05countmin}, CountSketch~\cite{charikarCFC04countsketch}, and varieties of Bloom filters~\cite{bloom70,GoodrichM11-ibt,MizrahiBLYR23-ibt-guarantees,FleischhackerLOS-ESA24-ibt-lessmem} are widely used in practice.

For network measurements, sketches have been used to track flow sizes~\cite{cormodeM05countmin,charikarCFC04countsketch}, heavy hitters~\cite{cormodeM05whatshot}, and flow size distributions~\cite{yang18elastic,kaicheng25towersketch}.
Sketches have also been used to track packet losses~\cite{LiMKY16-lossradar,LiMKY16-flowradar,ChameleMon23}.  See survey in~\cite{hanYJP22sketchnetworksurvey}.

\section{Problem statement}
\label{sect:prelim}

Here we describe the requirement for the sketch data structure, abstracted from the interface of FermatSketch from~\cite{ChameleMon23}, and review its implementation.  

We are given a multiset $F$ of flow ids $F=\{f_1,f_2,\ldots,f_n\}$ and we would like a data structure that reports a count statistics, i.e., it returns a list $(f,c_f)$ for each flow $f\in F$ where $c_f$ is the number of times $f$ appears in $F$.  

We also want the data structure to be linear.  As an example, consider a network configuration in Figure~\ref{fig:interconnect} where traffics come through link $A$, go though an interconnect, and exit through links $B, C,$ and $D$.  We can keep the flow count statistics for each link as a linear data structure: $S_A, S_B, S_C,$ and $S_D$.  
To find the lost flows, we can compute 
\[
S_A - (S_B + S_C + S_D),
\]
and extract the flow count statistics from this sketch.

We shall design a data structure $S$ that maintains a counter for each flow (using the flow id $f$ as a key).
It should support the following operations:
\begin{itemize}
\item $addFlow(S,f)$ -- adds $1$ to the counter for flow id $f$
\item $extract(S)$ -- returns the list of tuples $(f,c_f)$ where $f$ is the flow id and $c_f$ is the counter for flow id $f$.
\end{itemize}
We also want the data structure to be {\em linear}.  
More specifically, if we have data structures $S$ and $S'$ for input multisets $F$ and $F'$, 
the data structure for multiset $a\cdot F' +  b\cdot F$, where $a\cdot F$ represents a multiset whose elements from $F$ are duplicated for $a$ times, is just $a\cdot S'+b\cdot S$.  In particular, we can compute $S'-S$ such that
for each flow $f$ with counter $c$ in $S$ and $c'$ in $S'$ where $c\neq c'$, the sketch $S'-S$ keeps the counter of $c'-c$ for flow $f$.  
This is extremely crucial as the sketch is used to keep counter differences.

Since the data structure is probabilistic, we are interested in the case where we can successfully extract all the flow counters.

\subsection{Review of FermatSketch}

FermatSketch maintains $d$ equal-sized bucket arrays $\buck_1,\buck_2,\ldots,\buck_d$, each of size $m$.  
There is also a pairwise-independent hash function $h_i$ for each array $\buck_i$.
Each bucket $\buck_i[j]$ in the array $\buck_i$ has two components: the counter $\buck_i^{c}[j]$ and the ID sum $\buck_i^{id}[j]$.  The data structure also specifies a prime $p$, which is larger than any flow id's and any counter values.   Every computation on this sketch would be done modulo $p$; thus, each memory item in the data structure should be large enough to keep an integer as large as $p$.  We let $\Zp$ denote the set of integers modulo $p$.

We define a more general procedure $addFlow(S,f,c)$.  When it is called we do the following for each bucket array $i$, $i=1,\ldots,d$
\begin{itemize}
\item $\buck_i^c[h_i(f)] = (\buck_i^c[h_i(f)] + c) \bmod p$
\item $\buck_i^{id}[h_i(f)] = (\buck_i^{id}[h_i(f)] + c\cdot f) \bmod p$
\end{itemize}
When we perform this updated, we say that flow $f$ is {\em hashed} into this bucket $\buck_i[h_i(f)]$.
To implement $addFlow(S,f)$, we simply call $addFlow(S,f,1)$.

We say that bucket $\buck_i[j]$ is {\em pure} if there exists a unique flow $f$ that is hashed into $\buck_i[j]$
The key primitive operation for FermatSketch is the following pure bucket check.  
\begin{itemize}
\item $isPure(\buck_i[j])$ -- checks if $\buck_i[j]$ is pure..
\end{itemize}
Note that if the bucket is pure, i.e., there exists a single flow $f$ we should have that
\[
\buck_i^{id}[j] = f\cdot \buck_i^{c}[j]  \pmod p
\]
Solving this modular equation for $f$, we can obtain the candidate flow id $f$.  
Kaicheng {\em et al.}~\cite{ChameleMon23} uses the Fermat Little Theorem to solve this equation, i.e., they use theorem to find the modular multiplicative inverse of $\buck_i^{c}[j]$ modulo $p$; this inspires the name of the sketch.
After the candidate flow $f$ is obtained, we can verify that $f$ is the actual flow id by testing if $j=h_i(f)$, i.e., bucket $j$ is the possible bucket for flow $f$.  
This pureness test can be incorrect, i.e., a bucket may contain multiple flows, but the incorrect flow id $f$ obtained actually satisfies $h_i(f)=j$.  This occurs with probability $1/m$.  In~\cite{ChameleMon23}, the authors discussed this issue and stated a theorem (Theorem~3.1~in~\cite{ChameleMon23}) that if the total number of buckets $md$ is of the same order of the number $M$ of victim flows and $M$ is not too small, the extraction procedure succeeds with high probability.  They also describe a fingerprint verification (in Appendix A.4) to reduce the false positive rate, but this requires more memory.

To implement $extract(S)$, referred to as the decoding procedure in~\cite{ChameleMon23}, one could perform the pure bucket check for all non-zero buckets $\buck_i[j]$.  If the bucket is pure with flow $f$, we can take the count $c=\buck_i[j]$ and call $addFlow(S,f,-c)$ to ``remove'' flow $f$ from the sketch.  One could repeat this procedure until all buckets are zero or no pure buckets are left.  In the former case, we say that we have {\em successfully extracted} all the flows.

We remark that the FermatSketch is linear as all bucket components $\buck_i^c[j]$ and $\buck_i^{id}[j]$ are linear.

\section{The Improved Sketch Components}
\label{sect:algorithms}

In FermatSketch, pure buckets are starting points for flow extraction.  However, when the ratio betwen the number of flows and the number of buckets is relatively high, the number of non-pure buckets can get too large, leading to flow extraction failures.  Our key observation is that when the number of flows hashed into a bucket is a small constant (not necessarily $1$), we might still be able to recover all the flow ids.  

We start with an illustrative example.
When two flows $f_1$ and $f_2$ are hashed into the same bucket, we do not have enough information to recover their counters.  However, if we keeps 4 values:
\[
c_1 + c_2, \ \ \ \ 
c_1 + 2 c_2,
\]
\[
c_1\cdot f_1 + c_2\cdot f_2, \ \ \ \ 
c_1\cdot f_1 + 2 c_2 \cdot f_2,
\]
where $c_1$ and $c_2$ are counters for $f_1$ and $f_2$, we can solve for all the needed values: $f_1,f_2,c_1,$ and $c_2$.  With this idea, we hope that by keeping more values (as independent linear combinations) in each bucket, we can deal with higher numbers of hash collisions.  

With this advantage, by storing $k$ values in a bucket, we hope to extract $k$ flows from a bucket.
To keep the total space intact, we can reduce the number of buckets by a factor of $k$ as well.
Our intuition is that the increase in the average load factor would in turn reduce the ``variance'' on the number of pure buckets, resulting in the increase of the number of flows in pure buckets.
We perform a preliminary experiments to investigate this idea.  Section~\ref{sect:exp-prelim} presents the results.  

The challenge is how to make this idea works while ensuring that the data structure remains linear.

Our key idea is to add a ``random'' flow-dependent coefficient to each value added to the bucket.  
Let $k$ be a small constant, representing the maximum number of flows we can accept in a pure bucket.
We modify each bucket $\buck_i[j]$ to have $2k$ slots of two types:
\begin{itemize}
    \item the counters $\buck_i^c[j][1],\buck_i^c[j][2],\ldots,\buck_i^c[j][k]$
    \item the ID sums $\buck_i^{id}[j][1],\buck_i^{id}[j][2],\ldots,\buck_i^{id}[j][k]$
\end{itemize}
We also pick $k$ random hash functions $g_1,g_2,\ldots,g_k$ that map flow ids to $\Zp-\{0\}$, the set of non-zero integers less than $p$.  We also require that the range of each $g_i$ is of size $L$ and all ranges are disjoint.  The parameter $L$ will be determined later.

In our implementation, we let $g_i$ return values from a set of $L$ prime numbers in $\Zp$.  This is slightly different from the formal analysis in Section~\ref{sect:full-rank-analysis}.

When $addFlow(S,f,c)$ is called, we do the following for each bucket array $i$, $i=1,\ldots,d$, and for each index $r=1,\ldots,k$
\begin{itemize}
\item $\buck_i^c[h_i(f)][r] = (\buck_i^c[h_i(f)][r] + c\cdot g_r(f)) \bmod p$
\item $\buck_i^{id}[h_i(f)][r] = (\buck_i^{id}[h_i(f)][r] + c\cdot g_r(f)\cdot f) \bmod p$
\end{itemize}

With this modification, we shall analyze the values kept in a bucket $\buck_i[j]$ to devise a procedure for flow extraction.  
We start with a single bucket $\buck_i[j]$ with $k$ slots.  
Following~\cite{ChameleMon23}, we say that $\buck_i[j]$ is {\em pure} if there are at most $k$ flows hashed into it.  

Suppose that bucket $\buck_i[j]$ is pure and it has $k$ flows $f_1,f_2,\ldots,f_k$ hashed into it, with counters $c_1,c_2,\ldots,c_k$.  We first focus on the counters.
For each $r=1,\ldots,k$, we have
\[
\buck_i^c[j][r] = \sum_{t=1}^k c_t \cdot g_r(f_t) \pmod p
\]
Treating all $g_r(f_t)$'s as known, we have a linear system with $k$ equations and $k$ variables $c_1,c_2,\ldots,c_k$, as follows:
\[
\begin{bmatrix}
g_1(f_1) & g_1(f_2) & \cdots & g_1(f_k) \\
g_2(f_1) & g_2(f_2) & \cdots & g_2(f_k) \\
\vdots & \vdots & \ddots & \vdots \\
g_k(f_1) & g_k(f_2) & \cdots & g_k(f_k)
\end{bmatrix}
\begin{bmatrix}c_1 \\ c_2 \\ \vdots \\ c_k
\end{bmatrix}
\equiv
\begin{bmatrix}
b_1 \\
b_2 \\
\vdots \\
b_k
\end{bmatrix}\pmod p
\]
where $b_r=\buck_i^c[j][r]$ for $r=1,\ldots,k$.  The coefficient matrix, denoted by $M$, is a $k\times k$ matrix where the entry at row $r$ and column $t$ is $g_l(f_t)$.  If $M$ is full-rank, we can solve for all $c_t$'s.  
Similarly, we can also form a linear system to solve for all $f_t$'s using the ID sums $\buck_i^{id}[j][r]$'s, given that we already know all $c_t$'s.

To do this, first solve for ${\mathbf y}=[y_1,y_2,\ldots,y_k]^T$ such that
\[
M {\mathbf y} \equiv {\mathbf b'} \pmod p,
\]
where ${\mathbf b'}=[b'_1,b'_2,\ldots,b'_k]^T$ and $b'_r=\buck_i^{id}[j][r]$ for $r=1,\ldots,k$.  Since we already know all $c_t$'s, we can compute $f_t=y_t/c_t$ modulo $p$ for each $t=1,\ldots,k$.  Thus, the key to our approach is to ensure that the matrix $M$ is full-rank.  We shall analyze the probability of this event in Section~\ref{sect:full-rank-analysis}.

The catch is that we do not know the matrix $M$. 
However, we do know that each entry in row $r$ is from the range of $g_r$ which is of size $L$.  
Thus, we can try all possible choices for the matrix $M$ and solve the corresponding linear systems.  Since there are $L$ choices for each entry in $M$, there are $L^{k^2}$ possible choices for $M$.  For each choice of $M$, we can solve the two linear systems to obtain candidate values for $c_t$'s and $f_t$'s.  We can verify if these values are correct by checking if $j=h_i(f_t)$ for each $t=1,\ldots,k$.  If this holds, we can call $addFlow(S,f_t,-c_t)$ for each $t=1,\ldots,k$ to remove these flows from the sketch.  We can repeat this procedure until no more buckets can be extracted.

Procedure for testing bucket pureness and extracting $k$ flows from a pure bucket $\buck_i[j]$ is describe in Algorithm~\ref{fig:extract-k-flows}.

\begin{algorithm}
\caption{$extractBucket(S, \buck_i[j])$ \\ 
$\triangleright$ $extractBucket$ tests bucket pureness and extracts $k$ flows from a pure bucket $\buck_i[j]$.}
\label{fig:extract-k-flows}
\begin{algorithmic}
\State ${\mathcal M} \leftarrow$ all possible $k\times k$ matrices where the entry at row $r$ is from the range of $g_r$
\ForAll{$M\in{\mathcal M}$}
    \If{$M$ is not full-rank (modulo $p$)}
        \State continue
    \EndIf

    \State $\triangleright$ {\em solve the two linear systems}
    \State ${\mathbf b}\leftarrow [b_1,b_2,\ldots,b_k]^T$ where $b_r=\buck_i^c[j][r]$
    \State ${\mathbf b}'\leftarrow [b'_1,b'_2,\ldots,b'_k]^T$ where $b'_r=\buck_i^{id}[j][r]$
    \State Solve for ${\mathbf c}$ such that $M {\mathbf c} \equiv {\mathbf b} \pmod p$
    \State Solve for ${\mathbf y}$ such that $M {\mathbf y} \equiv {\mathbf b'} \pmod p$
    \For{$t=1$ to $k$}
        \If{$c_t\neq 0$}
            \State $f_t\leftarrow y_t/c_t \pmod p$
        \Else 
            \State $f_t\leftarrow null$
        \EndIf
    \EndFor

    \State $\triangleright$ {\em verify the solution}
    \If{$h_i(f_t)=j$ for $t=1,\ldots,k$ when $f_t\neq null$}
        \State \Return $\{(f_1,c_1),(f_2,c_2),\ldots,(f_k,c_k)\}$
    \EndIf
\EndFor
\State \Return ``fail''
\end{algorithmic}
\end{algorithm}

\section{Analysis}
\label{sect:analysis}
\subsection{Correctness}
There are two properties that we need to ensure.  
First, the linearity of the sketch.
This is fairly straight-forward since after all hash functions are chosen, all coefficients $g_i(f)$ is fixed for each flow $f$.  Thus, invoking $addFlow(S,f,c_1)$ after $addFlow(S,f,c_2)$ is equivalent to calling $addFlow(S,f,c_1+c_2)$, for any $c_1$ and $c_2$.

Second, we need to analyze the success probability of the extraction procedure.
This depends on two factors: (1) how flow hashing works, and given that a bucket contains at most $k$ flows, (2) the probability that the corresponding linear system is solvable.  We perform experiments to demonstrate the first factor in Section~\ref{sect:experiments}.  The second factor can be analyzed more regorously, as we shall do in the next subsection, Section~\ref{sect:full-rank-analysis}.  
We note that the solvability probability can be made arbitrarily close to $1$ by increasing $L$.  However, this drastically increases the running time.  Section~\ref{sect:running-time-analysis} provides the analysis.

\subsection{Analysis of the solvability probability}
\label{sect:full-rank-analysis}
We shall analyze the probability that the modular linear system from our algorithm is solvable.  Recall that we use modular arithematics and perform every operation modulo a large prime $p$.  Let $k$ denote that size of the matrix and also let $L$ be the number of coefficient choices.  

The random coefficients are from set of numbers $P_1,P_2,\ldots,P_k$ where each $P_i$ is a set of $L$ from $\Zp-\{0\}$, the set of non-zero integers less than $p$, and all $P_i$'s are disjoint.  Consider a $k\times k$ matrix $M$ where each entry on row $i$ is chosen uniformly from $P_i$.  

In our proof, we require a certain assumption related to the solution of a certain linear system modulo $p$ below.

\noindent
{\bf Assumption.}  The probability that a solution $\alpha_1,\alpha_2,\ldots,\alpha_j$ to a certian random linear system modulo $p$ with $j$ variables satisfies (1) at least two $\alpha_i$ and $\alpha_{i'}$ are nonzero and (2) $\alpha_1+\alpha_2+\cdots+\alpha_j \equiv 1$ modulo $p$ is neglegible.

\begin{lemma}
For $k>1$, under the above assumption,
the probability that $M$ is full-rank is at least $1-{k\choose 2}/L^k-2^k/(p-k^2)$.
\end{lemma}
\begin{proof}
Denote $k$ columns of $M$ as $C_1,C_2,\ldots,C_k$.  We first consider the probability that $M$ is singular due to two columns $C_i$ and $C_j$ being equal.  The probability that this occurs is $1/L^k$.   Since there are ${k\choose 2}$ pairs of columns, using the union bound, the probability that $M$ is singular due to this reason is at most ${k\choose 2}/L^k$.

Assuming that no two columns are equal, we now consider the probability that column $C_i$ is a linear combination of columns $C_1,C_2,\ldots,C_{i-1}$, given that columns $C_1,\ldots,C_{i-1}$ are linearly independent.  

We analyze a linear system over $\Zp$ modulo $p$ with $i-1$ variables $\alpha_1,\alpha_2,\ldots,\alpha_{i-1}$ such that
\[
\alpha_1 C_1 + \alpha_2 C_2 + \cdots + \alpha_{i-1} C_{i-1} \equiv C_i \pmod p
\]
Since there are $k$ rows and $i-1 < k$, the number of equations is more than the number of variables.  
Since $C_1,C_2,\ldots,C_{i-1}$ are linearly independent, there exists a set of rows $R=\{r_1,r_2,\ldots,r_{i-1}\}$ such that the submatrix $M'$ of $M$ formed by rows in $R$ and columns $C_1,C_2,\ldots,C_{i-1}$ is full-rank.  
Without loss of generality, we assume that $R=\{1,2,\ldots,i-1\}$, and define $C'_j$ to be the columns $C_j$ restricted to rows in $R$.  
The linear subsystem 
\[
\alpha_1 C'_1 + \alpha_2 C'_2 + \cdots + \alpha_{i-1} C'_{i-1} \equiv C'_i \pmod p
\]
with $i-1$ equations and $i-1$ variables has a unique solution $\alpha_1,\alpha_2,\ldots,\alpha_{i-1}$, since $C'_1,C'_2,\ldots,C'_{i-1}$ are linearly independent.

We now consider the probability that this solution also satisfies the remaining equations, especially the equation w.r.t the last row $k$, i.e., we want to analyze the probability that
\[
\alpha_1 C_1[k] + \alpha_2 C_2[k] + \cdots + \alpha_{i-1} C_{i-1}[k] \equiv C_i[k] \pmod p,
\]
where $C_j[k]$ is the entry of column $C_j$ at row $k$.  Recall that every entry $C_j[k]$ is chosen uniformly from the same set $P_k$.  

We first show that conditioning on the fact that at least one entry $C_j[k]$ is different from $C_i[k]$, the probability that the above equation holds is at most $2^{i-1}/(p-k^2-1)$.  For a nonempty subset $I\subseteq \{1,\ldots,i-1\}$, consider the case where $C_j[k]\neq C_i[k]$ for all $j\in I$, and $C_j[k]=C_i[k]$ for all $j\notin I$.  
Since $\alpha_j$'s are fixed, assume that all $\alpha_j$'s for $j\in I$ are already chosen, there is only one value for $C_i[k]$ that satisfies the above equation, i.e.,
\[
C_i[k] \equiv \frac{\sum_{j\in I} \alpha_j C_j[k]}{1-\sum_{j\not\in I}\alpha_j} \pmod p.
\]
Since $C_i[k]$ is chosen from a set of size at least $p-1-(k^2-1)=p-k^2$, the probability that it takes this value is at most $1/(p-k^2)$.  Taking the union bound over all subsets $I$, we have that the probability that the equation holds is at most $2^{i-1}/(p-k^2)$.

To see why we can ensure the conditioning, first of all, we can discard the case where exactly one $\alpha_j=1$ and all other $\alpha_{j'}=0$ for $j'\neq j$, since this would imply that $C_i=C_j$, contradicting our assumption.  
We are left with two cases, either there is exactly one nonzero $\alpha_j$ and $\alpha_j\neq 1$, or there are at least two nonzero $\alpha_j$'s.  
In the first case, since $\alpha_j\neq 1$, we know that $C_j[k]\neq C_i[k]$.  
In the second case, by our assumption, the probability that $\sum_{j}\alpha_j \equiv 1$ modulo $p$ is negligible; thus, we may assume that $\sum_{j}\alpha_j \not\equiv 1$ modulo $p$, implying that there must be at least one $j$ such that $C_j[k]\neq C_i[k]$ as needed.

Combining all bad events for columns $C_2,C_3,\ldots,C_k$, using the union bound, the probability that $M$ is singular due to this reason is at most $\sum_{i=2}^k 2^{i-1}/(p-k^2) < 2^k/(p-k^2)$, as desired.  
The lemma thus follows.
\end{proof}

We note that since $p$ is very large and $k$ is a very small constant, the term $2^k/(p-k^2)$ is also negligible; the key factor to solvability is $L$.  Section~\ref{sect:exp-prelim} presents some preliminary experiments to verify this analysis.

\subsection{Running time analysis}
\label{sect:running-time-analysis}

Assume that all hash functions are evaluated in $O(1)$ time. 
The running time for $addFlow(S,f,c)$ is $O(d\cdot k)$; since both numbers are small constants, this is effectively $O(1)$ time.

The running time for flow extraction is dominated by function $extractBucket()$, which solves linear systems of size $k\times k$ many times.
For a particular parameter $L$, there are at most $L^{k^2}$ choices for the coefficient matrix $M$.
Testing if it is full rank can be done in $O(k^3)$ time.  
In fact, we can find the inverse $M^{-1}$ modulo $p$ in that time as well; therefore, solving the two linear systems can be done later in $O(k^2)$ time.  Thus, the total running time for $extractBucket()$ is 
\[
O(k^3 L^{k^2}).
\]
To fully extract all flows, we need to call $extractBucket()$ successfully $md$ times.
By maintaining a list of successfully extracted buckets, and only calling $extractBucket()$ on affected buckets after each successful extraction, we can reduce the number of calls to $O(md)$ time.
Let denote the time function $extractBucket()$ is called as $N$.
In the worst case, the total running time for successful flow extraction is
\[
O(k^3 L^{k^2}\cdot N),
\]
where $N=O(md)$.  We remark that the factor of $N$ also appears in the original FermatSketch~\cite{ChameleMon23}, but since $k=1$, $extractBucket$ runs in $O(1)$ time.  Hence, the overhead of our approach is the additional multiplicative factor of $k^3 L^{k^2}$.

\section{Experiments}
\label{sect:experiments}

\subsection{Preliminary experiments}
\label{sect:exp-prelim}

We perform preliminary experiments to investigate the idea of using $k$ slots per bucket.
The first experiment verifies that increasing $k$ while reducing the number of buckets by the same factor to maintain the same space leads to improved percentage of flows in pure buckets.

We set the total number of slots to be $m\cdot k=6000$ and vary $k$ from $1$ to $4$.  
Note that the original FermatSketch corresponds to $k=1$.  We vary the number of flows $n$ from $300$ to $6000$, and measure the percentage of flows in pure buckets after all flows are added.  
The results shown in Figure~\ref{fig:exp-pure-percent-stat} are averaged over $100$ trials.
One can see clear advantages even when $k=2$.

\begin{figure}
\begin{center}
\includegraphics[width=0.4\textwidth]{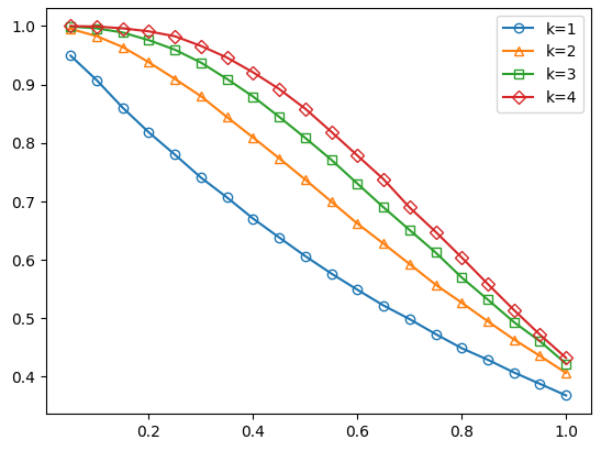}
\end{center}
\caption{The percentage of flows in pure buckets for $k=1,2,3,4$, for varying ratios between the number of flows and the total number of slots (i.e., $mk$ slots for $m$ buckets).}
\label{fig:exp-pure-percent-stat}
\end{figure}

Another issue we address is the solvability of the linear systems used in flow extraction, as analyzed in Section~\ref{sect:full-rank-analysis}.  As our goal is to get the probability of successful extraction very close to $1$, for each $k$, we vary $L$ to see the effect on the solvability probability.  However, using large $L$ leads to larger computational overhead as described in the running time analysis.  Therefore, we have to carefully choose $L$ to balance the success probability and the running time.
Figure~\ref{fig:exp-solvability-vary-l} shows the results for $k=2,3,4$  This experiment guides our choice of $L$ in the next set of experiments.

\begin{figure}
\begin{center}
\includegraphics[width=0.4\textwidth]{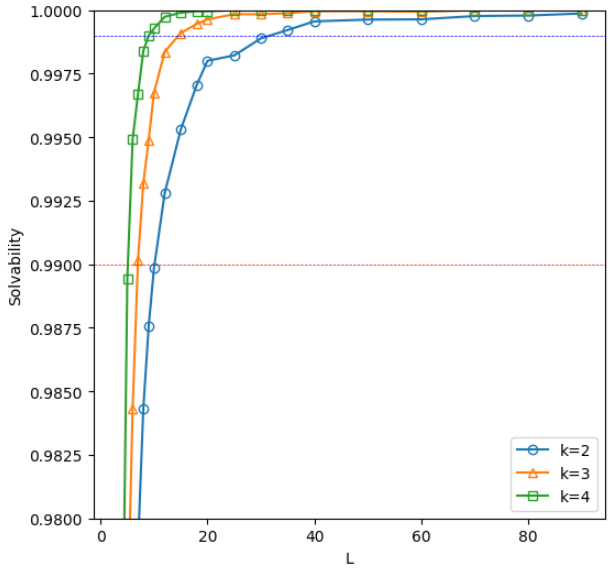}
\end{center}
\caption{The solvability probability for $k=2,3,4$, for varying values of $L$.}
\label{fig:exp-solvability-vary-l}
\end{figure}

\subsection{Successful extraction probability}

To demonstrate the advantages of using $k$ slots per bucket, we perform experiments to estimate the success extraction probabilities for the FermatSketch and for our new sketch design for $100$ victim flows.
We vary the number of slots $m'$ from $40$ to $60$, and for each value $m'$ we run FermatSketch with $m=m'$ buckets and for our sketch we set $k=2$ and $m=m'/2$ buckets; thus, both data structures use the same amount of space.  For our data structure, we vary $L$ from $6$ to $12$.
The results shown in Figure~\ref{fig:exp-success-vary-slots} are averaged over $2,000$ trials.

\begin{figure}
\begin{center}
\includegraphics[width=0.4\textwidth]{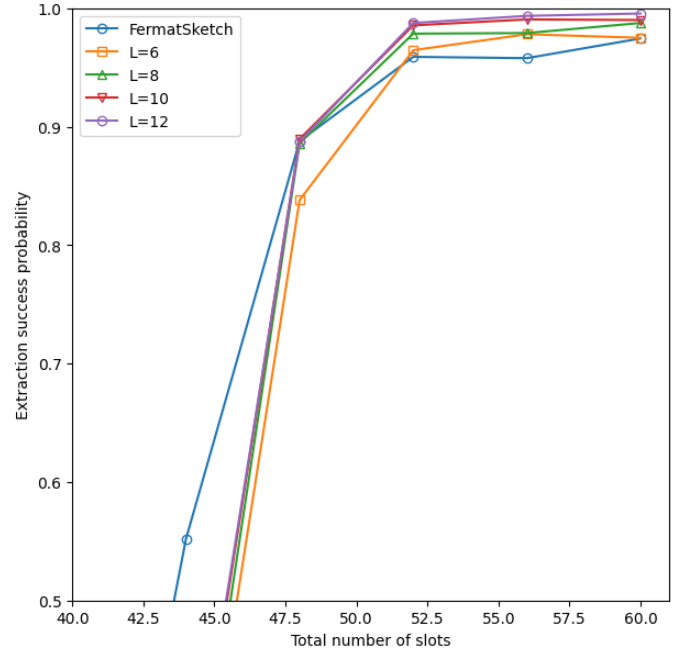}
\end{center}
\caption{The success extraction probability for $100$ victim flows for FermatSketch and our sketch with $k=2$ for varying values of $L\in\{6,8,10,12\}$.}
\label{fig:exp-success-vary-slots}
\end{figure}

\subsection{Required space for successful extraction}

In~\cite{ChameleMon23}, FermatSketch is used as a subroutine to the ChameleMon system.
Kaicheng~{\em et al.}~\cite{ChameleMon23} performed experiments to determine the required space for successful extraction of all flows with $99.9\%$ probability.  Because of limited computation resource, in this current work we perform experiments with $99\%$ success probability.  

We perform experiments to determine the required number of slots (i.e., the number of buckets in the FermatSketch and $k$ times the number of buckets in our case).  We run our experiments for $50, 100, 150, 200, 300,$ and $500$ victim flows.
To find the number of slots, we perform a binary search, and for each guess, we run $1,000$ trials to estimate the success probability.  Figure~\ref{fig:exp-required-slots} shows the results.
Our best choice of $L$ is $10$ and this gives us roughly $34.9\%$ reduction in space when the number of victim flows is small ($100$), but the improvement drops to $8\%$ when the number of flows increases.

We note that we use the provided the implementation of FermatSketch available at
{\tt\small https://github.com/ChameleMoncode/ChameleMon}.  The hash familiy used is MurmurHash3~\cite{senumaMmh3PythonExtension2025}.


\begin{figure}
\begin{center}
\includegraphics[width=0.4\textwidth]{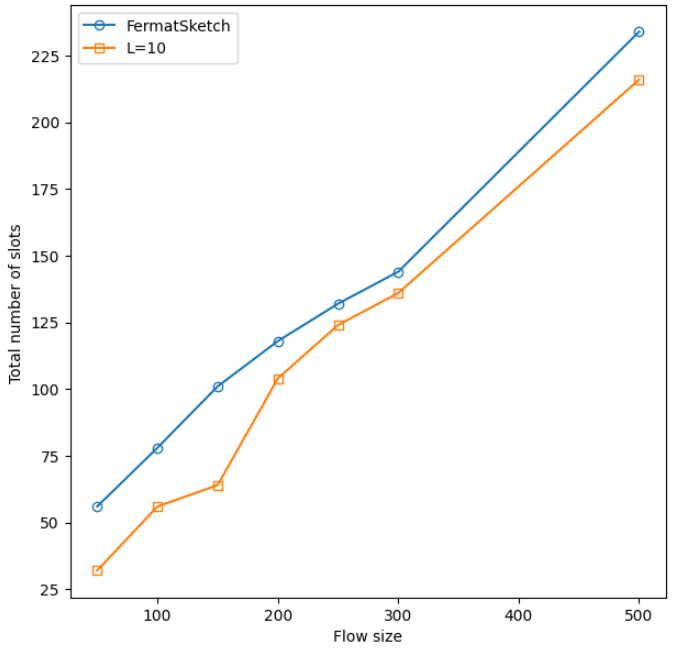}
\end{center}
\caption{The required number of slots for successful extraction for varying numbers of victim flows for FermatSketch and our sketch with $k=2$ and $L=10$.}
\label{fig:exp-required-slots}
\end{figure}

\bibliographystyle{IEEEtran}

\bibliography{sketches}

\end{document}